\documentclass[a4paper]{article}
\usepackage[utf8]{inputenc}
\usepackage{fullpage}
\usepackage{amsthm}
\usepackage{amsfonts}
\usepackage{amsmath}
\usepackage{thm-restate}
\usepackage{todonotes}
\usepackage[hypertexnames=false,colorlinks=true,urlcolor=blue,citecolor=green,linkcolor=red]{hyperref}
\usepackage{cleveref}

\newtheorem{theorem}{Theorem}[section]
\newtheorem{definition}[theorem]{Definition}

\newtheorem{observation}[theorem]{Observation}
\newtheorem{corollary}[theorem]{Corollary}
\newtheorem{lemma}[theorem]{Lemma}
\newtheorem{claim}[theorem]{Claim}

\title{Optimal Vertex-Cut Sparsification of Quasi-Bipartite Graphs}

\iffalse
\author[1]{Itai Boneh \thanks{Supported by Israel Science Foundation grant \#1475/18}}
\author[2]{Robert Krauthgamer \thanks{Supported in part by ONR Award N00014-18-1-2364, Israel Science Foundation grant \#1086/18, the Weizmann Data Science Research Center, and the Minerva Foundation. }}
  \affil[1]{Bar-Ilan University, Israel \\ \href{mailto:itai.bone@biu.ac.il}{itai.bone@biu.ac.il} }
  \affil[2]{Weizmann Institute of Science, Israel \\
  \href{mailto:robert.krauthgamer@weizmann.ac.il}{robert.krauthgamer@weizmann.ac.il}}
\else
\author{Itai Boneh%
  \thanks{Supported by Israel Science Foundation grant \#1475/18}
  \\ Bar-Ilan University
  \\ \texttt{itai.bone@biu.ac.il}
  \and
  Robert Krauthgamer%
  \thanks{Supported in part by ONR Award N00014-18-1-2364, Israel Science Foundation grant \#1086/18, the Weizmann Data Science Research Center, and the Minerva Foundation. }
  \\ Weizmann Institute of Science
  \\ \texttt{robert.krauthgamer@weizmann.ac.il}
}
\fi

\def\NN{{\mathbb{N}}}
\DeclareMathOperator{\mincut}{mincut}
\DeclareMathOperator{\poly}{poly}
\newcommand{\set}[1]{\{#1\}}
\newcommand{\dd}{\mathinner{.\,.}}
\newcommand{\minn}[1]{\min\{{#1}\}}

\newcommand{\NewSet}{M^+}
\newcommand{\OldSet}{M^-}

\begin{document}

\maketitle

\begin{abstract}
In vertex-cut sparsification,
given a graph $G=(V,E)$ with a terminal set $T\subseteq V$, 
we wish to construct a graph $G'=(V',E')$ with $T\subseteq V'$,
such that for every two sets of terminals $A,B\subseteq T$,
the size of a minimum $(A,B)$-vertex-cut in $G'$ is the same as in $G$.
In the most basic setting, $G$ is unweighted and undirected, 
and we wish to bound the size of $G'$ by a function of $k=|T|$.
Kratsch and Wahlström [JACM 2020] proved that every graph $G$ (possibly directed),
admits a vertex-cut sparsifier $G'$ with $O(k^3)$ vertices,
which can in fact be constructed in randomized polynomial time.

We study (possibly directed) graphs $G$ that are quasi-bipartite,
i.e., every edge has at least one endpoint in $T$,
and prove that they admit a vertex-cut sparsifier with $O(k^2)$ edges and vertices,
which can in fact be constructed in deterministic polynomial time. 
In fact, this bound naturally extends to all graphs with a small separator
into bounded-size sets.
Finally, we prove information-theoretically a nearly-matching lower bound,
i.e., that $\tilde{\Omega}(k^2)$ edges are required to sparsify
quasi-bipartite undirected graphs.
\end{abstract}

\newpage

\setcounter{page}{1}
\section{Introduction}

Vertex sparsification is a genre of problems,
where given a graph $G=(V,E)$ and a set of vertices $T\subseteq V$ called \emph{terminals},
the goal is to find a small graph $G'=(V',E')$, called a \emph{sparsifier},
that includes the terminals (i.e., $T\subseteq V'$)
and maintains certain properties that the terminals have in $G$. 
Usually, one aims at sparsifier size that is bounded by a function of $k=|T|$,
e.g., $|V'| \leq \poly(k)$.
Several properties have been studied in this context,
including distances between every two terminals
\cite{Gupta01, KNZ14}, %
minimum edge cuts (between two terminals or two sets of terminals) 
\cite{GH61, HKNR98, ChalermsookDKLL21}, %
minimum vertex cuts \cite{KW20,HeLW21}, %
multicommodity flow
\cite{Moitra09,Chuzhoy12,AGK14}, %
and effective resistances \cite{DGGP19}
(we provide here only a few example references, a comprehensive list would be excessive). 

Sparsification is a natural method to compress a graph,
in the sense of reducing the size of its representation,
which can be very effective when storing or communicating it. 
Computing a sparsifier can also be used as a preprocessing step
before executing some algorithm;
the idea is that reducing the input size
will decrease the running time of the ``main'' algorithm,
and this further requires a fast construction of the sparsifier. 
The study of vertex sparsification can be divided roughly into two challenges: 
\emph{Combinatorially}, do sparsifiers of certain size exist at all, 
e.g., can the size bound depend only on $k$ and if so, what is the best such bound (for all graphs $G$ or for a family of graphs)? 
And \emph{computationally},
how fast can one construct a sparsifier for an input graph $G$?

A well-known example is a \emph{mimicking network},
which in the above language is a vertex sparsifier
that maintains \emph{exactly} the minimum edge cuts between every two sets of terminals.%
\footnote{There is also a long line of work on sparsifiers that maintain
  these minimum edge cuts \emph{approximately}, see e.g.~\cite{Moitra09,Chuzhoy12,AGK14}.
}
It was introduced by Hagerup, Katajainen, Nishimura and Ragde \cite{HKNR98},
who provided a sparsifier construction with $2^{2^k}$ vertices.
Their upper bound was slightly improved by Khan and Raghavendra~\cite{KR14}. 
Better mimicking networks, i.e., sparsifiers of smaller size, were constructed 
for graphs with bounded treewidth \cite{CSWZ00} %
and for planar graphs \cite{KR13,KR20}, %
and some lower bounds are also known \cite{KR13,KR14,KPZ17}.

We study a different but related notion
of sparsification that maintains \emph{minimum vertex cuts},
so let us recall its basic terminology. 
A vertex cut in $G$ between two sets of vertices $A,B \subseteq V$,
also called an $(A,B)$-vertex cut,
is a set of vertices $C \subseteq V$ whose removal from $G$
eliminates all paths from $A$ to $B$.
Note that $C$ may intersect $A\cup B$,
in fact our definition above allows $A$ and $B$ to intersect, 
and then clearly $A\cap B\subseteq C$.
A \emph{vertex-cut sparsifier} is a graph $G'$ that maintains,
for every two subsets of terminals $A,B\subseteq T$,
the minimum size of an $(A,B)$-vertex-cut in $G$. 
Observe that these definitions extend immediately to directed graphs.

The state-of-the-art solution for vertex-cut sparsification of a general digraph $G$ 
is a randomized algorithm of Kratsch and Wahlström~\cite{KW20} that,
given a digraph with $k$ terminals,
constructs in polynomial time a vertex-cut sparsifier $G'$ with $O(k^3)$ vertices.
They devised a powerful new technique
of iteratively removing an ``irrelevant'' vertex in the graph,
which guarantees that the removal does not affect any minimum vertex cut.
The irrelevant vertex is identified by
computing some $O(k^3)$-size set, and showing that every relevant vertex must correspond to a distinct element in that set.
The main innovation in their solution is finding the $O(k^3)$-size set using tools from matroid theory.
They also provided a lower bound by presenting (a family of) directed graphs
with $k$ terminals,
for which every sparsifier must have $\Omega(k^2)$ vertices.
For the special case of directed acyclic graphs, 
the upper bound $O(k^3)$ was recently improved to $O(k^2)$ vertices
by He, Li, and Wahlström~\cite{HeLW21},
using the techniques of \cite{KW20} and additional ideas. 
They also proved that $\Omega(k^2)$ vertices are required to sparsify
directed acyclic graphs.

This sparsification of Kratsch and Wahlström~\cite{KW20}
was motivated by kernelization, an important notion in parameterized complexity,
where an input is preprocessed in polynomial-time to reduce its size
while maintaining the optimal value of some optimization problem,
e.g, some cut problem. 
Indeed, some of their main results are kernels of polynomial size 
(i.e., polynomial in the number of terminals $k$) for several problems.
The sparsification results of \cite{KW20} have strong implications
for two other vertex-sparsification problems:
One is sparsification of unweighted graphs
that maintains the minimum edge cuts between every two sets of terminals. 
Chuzhoy~\cite{Chuzhoy12} designed such sparsifiers
that maintain these minimum edge cuts within factor $O(1)$
and have size $O(Z^3)$, %
where $Z$ is the sum of degrees of the terminals in the input graph.
A simple application of~\cite{KW20} yields sparsifiers 
that maintain the minimum edge cuts exactly,
and moreover it improves the sparsifier's construction time
(from exponential to polynomial in $Z$).
The sparsification results of~\cite{KW20} were used also for sparsifiers
that maintain the minimum edge cuts between every two sets of terminals
up to threshold $c$ \cite{ChalermsookDKLL21,abs-2011-15101}. 
Additionally, the techniques developed in~\cite{KW20} 
have been utilized to obtain kernels for other NP-hard problems,
see e.g.~\cite{DBLP:journals/corr/HolsK15,DBLP:journals/siamdm/Kratsch18}.

\subsection{Our Results}

We study vertex-cut sparsifiers for graphs that are \emph{quasi-bipartite},
meaning that every edge has at least one of its endpoints in $T$.
We design three sparsifier constructions,
all presented in Section \ref{sec:sparsifier}.
Our first and main result is that quasi-bipartite graphs with $k$ terminals
admit vertex-cut sparsifiers with $O(k^2)$ edges and vertices;
moreover, these sparsifiers can be constructed efficiently.
Our construction does not rely on matroids and representative sets,
thus offering new insights and more elementary techniques
for constructing vertex-cut sparsification.

\begin{theorem}\label{t:qbsparsifier}
Every quasi-bipartite directed graph $G=(V,E)$ with $k$ terminals 
admits a vertex-cut sparsifier $G'=(V',E')$ with $|E'| = O(k^2)$. 
Furthermore, given $G$ and $T$, 
such a sparsifier can be computed in deterministic polynomial time. 
\end{theorem}

Observe that a graph $G$ with terminals $T$ is quasi-bipartite
if and only if deleting the terminals from the graph
leaves only isolated vertices, i.e., all connected components have size $1$.
From this viewpoint, it is natural to generalize our result to inputs $(G,T)$
in which all connected components of $G\setminus T$ have bounded size,
as follows.

\begin{theorem}\label{t:qbsparsifiergen}
Every directed graph $G=(V,E)$ with $k$ terminal vertices $T \subseteq V$
admits a vertex-cut sparsifier $G'=(V',E')$ with size bounds 
$|V'| = O(c k^2)$ and $|E'| = O((c k)^2)$,
where $c=c_G$ is the maximum number of vertices in a connected component of $G\setminus T$. 
Furthermore, given $G$ and $T$, 
such a sparsifier can be computed in deterministic polynomial time. 
\end{theorem}

We can generalize this result even further, 
to graphs that have a small separator as in Defenition \ref{def:sep},
which informally says that one can delete a few vertices (at most $x$) 
so as to leave connected components all of bounded size (at most $\tau$). 

\begin{corollary}\label{cor:sparsmallsep}
Every directed graph $G=(V,E)$ with $k$ terminal vertices $T \subseteq V$
that has a $\tau$-separator (see Definition \ref{def:sep}) of size $x$, 
admits a vertex-cut sparsifier $G'=(V',E')$ with size bounds 
$|V'| = O(\tau(k+x)^2)$ and $|E'| = O(\tau^2(k+x)^2)$.
Furthermore, given $G$ and $T$, one can compute in deterministic polynomial time
a sparsifier $\Tilde{G} = (\Tilde{V},\Tilde{E})$ 
with $|\Tilde{V}| = O(\tau (k +x \tau)^2)$ 
and $|\Tilde{E}| = O(\tau^2(k+x \tau)^2)$. 
\end{corollary}

The $\Omega(k^2)$ lower bound of Kratsch and Wahlstöm \cite{KW20}
is actually proved for a quasi-bipartite graph $G$,
and therefore our sparsifier construction is optimal
for \emph{directed} quasi-bipartite graphs.
We extend their lower bound to \emph{undirected} quasi-bipartite graphs, 
albeit at a loss of a logarithmic factor (and using different techniques).
This shows that our sparsifier constructions are near-optimal even for undirected inputs. 
We actually prove in Section \ref{sec:lower} two lower bounds.
The first one holds for a sparsifier $G'$ that must be a subgraph of the input graph $G$,
which is consistent with our sparsifier construction in Theorem \ref{t:qbsparsifier}. 
The second lower bound holds for every sparsifier,
and uses information-theoretic technique. 

\begin{theorem}\label{t:subgraphlowerbound}
(See Theorem \ref{t:subgraphlowerboundweighted}.)  
For every $k\in \NN $, there is an undirected quasi-bipartite graph $G$ with $k$ terminals, 
such that every vertex-cut sparsifier of $G$ which is a subgraph of $G$ 
must have $\Omega(k^2)$ vertices.
\end{theorem}

It follows that our sparsifier in Theorem \ref{t:qbsparsifier},
which is a subgraph and has $O(k^2)$ edges (and vertices),
is tight, i.e., achieves an optimal bound,
at least when using the technique of subgraph sparsification. Note that Theorem \ref{t:subgraphlowerbound} is not derived from the $\Omega(k^2)$ vertices lower bound presented by \cite{KW20}, as their statement relates to directed graphs.

We also present an information-theoretic lower bound on the representation size of a sparsifier of quasi-bipartite graphs.
It directly leads to the following result.

\begin{theorem}\label{t:infolowerbound}
(See Theorem \ref{t:lowergeneral}.)
For every $k\in \NN $, there is an undirected quasi-bipartite graph $G$ with $k$ terminals, 
such that every vertex-cut sparsifier of $G$ 
must have $\tilde{\Omega}(k^2)$ edges.
\end{theorem}

\section{Preliminaries}
\label{sec:prelims}

\begin{definition}[Vertex Cut]
Let $G=(V,E)$ be an unweighted directed (resp. undirected) graph.
A \emph{vertex cut} between two subsets $A,B \subseteq V$, 
called in short an \emph{$(A,B)$-vertex-cut},
is a subset $C\subseteq V$ whose removal disconnected $A$ from $B$,
i.e., for all $a\in A \setminus C,b\in B \setminus C$ there is no directed (resp. undirected) path
from $a$ to $b$ in the graph $G\setminus C$. 
\end{definition}

We denote by $\mincut_G(A,B)$
the minimum size of a vertex cut between $A$ and $B$ in $G$. 
Note that our definition above does not require an $(A,B)$-vertex-cut to be disjoint from $A\cup B$.
In particular, $A$ and $B$ are themselves valid $(A,B)$-vertex-cuts,
and thus $\mincut_G(A,B)\le \minn{|A|,|B|}$.

\begin{definition}[Vertex-Cut Sparsifier]
\label{def:vertexcutsparsifier}
Let $G=(V,E)$ be an unweighted directed (resp. undirected) graph.
A \emph{vertex-cut sparsifier of $G$}
with respect to a set of terminals $T\subseteq V$ 
is a graph $G'=(V',E')$ that contains the terminals, i.e., $T\subseteq V'$,
and
\[
  \forall A,B\subseteq T, 
  \quad
  \mincut_G(A,B) = \mincut_{G'}(A,B). 
\]
\end{definition}

Note that for every $A,B \subseteq T$ with $A\cap B = D$, $A' = A\setminus D$, $B' = B\setminus D$ and $T' = T\setminus D$ we have 
\[
\mincut_G(A,B) = \min \set{\mincut_{G\setminus D}(A^*,B^*) \mid A' \subseteq A* \subseteq T'\setminus B', B^* = T'\setminus A^*} 
\]
The above equality holds because the vertices in $D$ are forced to be in $C$. Then, the superset $A^*$ can be interpreted as the set of vertices connected to $A'$ in $G\setminus C$ for a minimum $(A,B)$-vertex cut $C$.

It follows that it is sufficient to demand that for every $A,D \subseteq T$ disjoint subsets of $T$, $G'$ has $\mincut_{G\setminus D}(A,T\setminus (A\cup D)) = \mincut_{G'\setminus D}(A,T\setminus (A\cup D))$ for $G'$ to be a vertex cut sparsifier of $(G,T)$.

\begin{definition}[Quasi-Bipartite Graph]
A graph $G=(V,E)$ with terminals $T\subseteq V$
is called \emph{quasi-bipartite}
if every edge $e\in E$ has at least one endpoint in $T$. 
\end{definition}

In a directed graph $G=(V,E)$, two vertices $u,v \in V$ are in the same \textit{weakly connected} component if there is a path from $u$ to $v$ ignoring the directions of the edges of $G$.

\begin{definition}[$\tau$-separator, $\tau$-Quasi-Bipartite Graph]
\label{def:sep}
Let $G=(V,E)$ be an undirectred (resp. directed) graph.
We say that a vertex subset $S\subseteq V$ is a \emph{$\tau$-separator} of $G$
if every connected component (resp. weakly connected component) in $G\setminus S$
is of size at most $\tau$.

A graph $G$ with terminals $T$ is called \emph{$\tau$-quasi-bipartite}
if $T$ is a $\tau$-separator of $G$.
\end{definition}

Notice that quasi-bipartite is equivalent to $1$-quasi-bipartite;
hence, the family of $\tau$-quasi-bipartite graphs generalizes that of quasi-bipartite graphs.

\section{Sparsification Algorithms}
\label{sec:sparsifier}

In this section, we start by restricting our attention to quasi-bipartite graphs. We later show how to generalize our construction to sparsify $\tau$-quasi bipartite graphs.

\subsection{Sparsifiers with $O(k^2)$ Edges for Quasi-Bipartite Graphs}
\label{sec:sparsifierSimple}

We start by presenting a sparsifier construction
for an undirected quasi-bipartite graph $G=(V,E)$,
where $T\subseteq V$ is the set of terminals
and $N := V\setminus T$ is the set of non-terminals.
Since the graph is undirected, we shall denote edges as unordered pairs, e.g., $\set{a,b}$. 
Without loss of generality,
we may assume throughout that $G$ is a bipartite graph with sides $T$ and $V$.
Indeed, every edge $e=\set{a,b}$ that connects two terminals $a,b\in T$
can be subdivided, using a new non-terminal $v_e$, into two edges $\set{a,v_e},\set{v_e,b}$. 
It is easily verified that this step does not modify
the value of any relevant vertex cut (between subsets $A,B\subseteq T$).

\begin{definition}[Linking Edge, Link Graph]
We start by considering an undirected quasi-bipartite graph $G=(T\cup N,E)$. We will later show how to apply our construction to a directed quasi-bipartite graph. 
We say that an edge $e=\{a,v\}\in E$ 
\emph{links} terminal $a\in T$ to terminal $b\in T$
if both $\{a,v\}, \{v,b\} \in E$.%
\footnote{Informally, this is just the first edge on a length-2 path from $a$ to $b$.
Notice that we treat $(a,b)$ as an ordered pair here.
}
The \emph{link graph} of $G$ is the bipartite graph $\mathcal{L}_G$
whose vertex set has two sides $T\times T$ and $E$
and its edge set is $E_{\mathcal{L}} = \set{\set{(a,b),e} \mid \text{$e$ links $a$ and $b$} }$.
\end{definition}

With that, we are ready to present our construction.
Given an input graph $G=(V,N,E)$,
construct its link graph $\mathcal{L}_G$
and compute in it a maximum matching $M$.
Then construct the sparsifier $G'$ as follows.
For every matching edge $\{ (a,b),\{a,v\} \}\in M$, 
include in $G'$ the edges $\{a,v\}$ and $\{v,b\}$,
and the corresponding vertex $v$.
(If this rule includes the same edge or vertex multiple times,
it will appear in $G'$ only once.)
Formally, $G'=(V',E')$ is given by:
\begin{enumerate}
\item $V' := T \cup \bigcup_{\{(a,b),\{a,v\}\} \in M} \{ v \}$.
\item $E' := \bigcup_{\{(a,b),\{a,v\}\} \in M} \{ \{a,v\}, \{ v,b\} \}$
\end{enumerate}  
It is clear that $G'$ can be computed in $O(|E|)$ time (ignoring $poly(|T|)$ factors),
and that $|E'| \leq 2|M| = O(|T|^2)$.

\begin{lemma} \label{lem:sparsifierSimple}
$G'$ is a vertex-cut sparsifier of $G$.
\end{lemma}

\begin{proof}
Assume to the contrary that $G'$ is not a vertex-cut sparsifier. It follows that there are two disjiont sets of terminals $A,D\subseteq T$ such that $\mincut_{G\setminus D}(A,B) \neq \mincut_{G'\setminus D}$ with $B=T\setminus ( A \cup D)$ (According to the discussion following Definition \ref{def:vertexcutsparsifier}). For the remainder of this proof, we assume $D = \emptyset$. Removing this assumption does not require any significant modification to our proof - a proof without this assumption is simply obtained by replacing every instance of $G$ (resp. $G'$) in our proof with $G\setminus D$ (resp. $G' \setminus D$).

Since $G'$ is a subgraph of $G$, 
every $(A,B)$-vertex-cut in $G$ is also an $(A,B)$-vertex-cut in $G'$, 
and thus $\mincut_{G'}(A,B) < \mincut_{G}(A,B)$. 

Let $C\subseteq V$ be a minimum $(A,B)$-vertex-cut in $G'$ 
that has a \emph{maximal} number of terminals. Then $|C| < \mincut_{G}(A,B)$, 
and thus $G\setminus C$ contains
two terminals $a \in A \setminus C$ and $b^* \in B\setminus C$ 
that are connected by a path $P=(a,\ldots, b^*)$. 
Since $G$ is bipartite, every non-terminal in the path $P$ is followed by a terminal.
We can therefore assume without loss of generality (by exchanging $a,b^*$) 
that $P=(a,v_0,b^*)$ for some non-terminal $v_0 \notin C$. 
Since $C$ is an $(A,B)$-vertex-cut in $G'$, 
at least one of $\{a,v_0\},\{v_0,b^*\}$ is not an edge in $G'$. 
We assume without loss of generality that the edge $e_0 = \{ a,v_0 \}$ is missing from $G'$. 
Notice that $e_0$ links $a$ and $b^*$ in $G$, 
and its absence from $G'$ indicates that $e_0$ is not matched by $M$. 

Both $A$ and $B$ are $(A,B)$-vertex-cuts in $G'$, thus
$|C| \le \min(|A|,|B|) < t$ for $t:= \lceil \frac{|T|}{2} \rceil+1$. 
We next show that there is a sequence $v_0,v_1,\ldots,v_t$ of distinct vertices,
such that $v_i\in C$ for all $i\in [1\dd t]$,
and thus reach a contradiction that $|C| \geq t$. 

Formally, we construct two sequences $S_v = (v_0,v_1,\ldots,v_t)$ and $S_b = (b_1,b_2, \ldots, b_t)$ of \emph{distinct} vertices 
and a decreasing function $f: [1\dd t] \to [0\dd t]$,
that satisfy the following four invariants.
\begin{enumerate}
    \item \label{seqinvar:edges} For every $i\in [0\dd t]$, there is an edge $e_i = \{a,v_i\} \in E$.
    \item \label{seqinvar:inC} For every $i\in [1\dd t]$, $v_i \in C$.
    \item \label{seqinvar:inB} For every $i\in [1\dd t]$, $b_i \in B\setminus C$, and $((a,b_i),e_i) \in M$.
    \item \label{seqinvar:func} For every $i\in [1\dd t]$, $e_{f(i)}$ links $a$ and $b_i$.
\end{enumerate}

Our construction of $S_v$, $S_b$ and $f$ is by induction on $i=0,1,2,\ldots$,
namely, each step $i$
constructs the prefix $S^i_v = (v_0,\ldots, v_{i})$ of $S_v$
and the prefix $S^i_b = (b_1,\ldots, b_{i})$ of $S_b$,
and also determines values $f(x)$ for $x\in [1\dd i]$, 
in a manner that satisfies the four invariants. 

In the base case $i=0$, we initialize $S^0_v = (v_0)$ and $S^0_b$ to be an empty sequence, and no value is decided for $f$.
All the invariants are satisfied
(invariants~\ref{seqinvar:inC}-\ref{seqinvar:func} vacuously).

For $i=1$, we extend the prefixes as follows.
Recall that $e_0 = \{a,v_0\}$ is not matched by $M$.
If the terminal pair $(a,b^*)$ is not matched by $M$,
then $M' = M\cup \{ ((a,b^*),e_0)\}$ is a matching in $\mathcal{L}_G$,
which contradicts the maximality of $M$.
It follows that $(a,b^*)$ is matched by $M$,
i.e., $\{(a,b^*),e_1\} \in M$ for some $e_1=\{a,v^*\}$, and according to our construction, $\set{a,v^*},\set{v^*,b} \in E'$. Since $C$ is an $(A,B)$ cut in $G'$, and $a,b^* \notin C$, we must have $v^* \in C$.
We can therefore set $v_1 = v^*$, $b_1 = b^*$, and $f(1) = 0$ (see  Figure \ref{fig:step0}) to satisfy Invariants \ref{seqinvar:edges}-\ref{seqinvar:func}.

We proceed to the case $i\in [2\dd t]$.
Let $V_i = \{v_1,\ldots, v_{i-1}\}$, $B_i = \{b_1,\ldots,b_{i-1}\}$,
and $C' = C \cup B_i \setminus V_i$. 
Observe that $|C'| = |C|$ and $C'$ contains $i-1\ge 1$ more terminals than $C$,
hence $C'$ cannot be an $(A,B)$-vertex-cut in $G'$. 
It follows that $G'\setminus C'$ contains a path $P' = (a',u,b^*)$
from terminal $a'\in A \setminus C'$ to terminal $b^* \in B\setminus C'$
through non-terminal $u \in N\setminus C'$.

Since $C$ is an $(A,B)$-vertex-cut in $G'$,
it must contain at least one of the vertices in the path $P'$.
But the terminals $a'$ and $b^*$ cannot be in $C$,
because the terminals in $C$ are contained also in $C'$,
and therefore $u\in C\setminus C' = V_i$, i.e., $u = v_j$ for some $j<i$. 
Notice that the edge $e_j =\set{ a,v_j}$ links $a$ and $b^*$.
Assume for now that the pair $(a,b^*)$ is matched by $M$ to some $e^*=\set{a,v^*}$.
Under this assumption, we set $v_i = v^*$, $b_i = b^*$, and $f(i) = j$,
and we need to show that this assignment satisfies all the invariants.

First, we need to show that vertices $b^*$ and $v^*$ are distinct from their respective sequences.
We know $b^* \notin C'$ and thus $b^* \notin B_i$.
Assume towards contradiction that $v^* = v_x$ for some $x < i$. It follows that $e' = \set{a,v_x} = e_x$ is matched to $(a,b_i)$ in $M$.
By the inductive hypothesis about prefixes $S^{i-1}_v$ and $S^{i-1}_b$,
the edge $e_x$ is matched with $(a,b_x) \neq (a,b_i)$, reaching a contradiction.

Invariants~\ref{seqinvar:edges}, \ref{seqinvar:inB},
and~\ref{seqinvar:func} are clearly satisfied.
Since $(a,b_i)$ is matched with $e' = \set{a,v^*}$ in $M$,
the edges of the path $\tilde{P}=(a,v^*,b^*)$ are in $G'$
and therefore one of the vertices of $\tilde{P}$ must be in the $(A,B)$-vertex-cut $C$.
Since $a,b^* \notin C$, we have that $v^*\in C$ and Invariant~\ref{seqinvar:inC} is satisfied as well.

We have thus shown that if $(a,b^*)$ is matched in $M$
then the prefixes can be properly extended.
We proceed to prove that this is always the case.
\begin{claim}\label{c:abprimematched}
$(a,b^*)$ is matched in $M$.
\end{claim}
\begin{proof}
Assume towards contradiction that $(a,b^*)$ is not matched in $M$.
Denote $b^*=b_i$. We set $f(i) = j$ and denote as $y \ge 1$ the minimal integer such that $f^y(i) = 0$ (where $f^y$ denotes applying $f$ repeatedly $y$ times).
Since $f$ is decreasing, $y$ is well defined. We denote $f^0(i) = i$.

Recall that $e_j$ links $a$ and $b_i$.
Therefore, the inductive hypothesis implies that $e_{f(x)}$ links $a$ and $b_x$ for all $x\in[1\dd i]$ (Invariant \ref{seqinvar:func}).

We define the following sets of edges in $\mathcal{L}_G$: 
\begin{enumerate}
\item $\NewSet = \Big\{ \set{(a,b_i),e_{f(i)}},\set{(a,b_{f(i)}),e_{f^2(i)}}, \ldots, \set{ (a,b_{f^{y-1}(i)}) , e_0} \Big\}$ 
\item $\OldSet = \Big\{ \set{(a,b_{f(i)}) , e_{f(i)}},\set{(a,b_{f^2(i)}),e_{f^2(i)}}, \ldots, \set{(a,b_{f^{y-1}(i)}),e_{f^{y-1}(i)}} \Big\}$
\end{enumerate}

We proceed to show that $M' = M \cup \NewSet \setminus \OldSet$
is a matching of size $|M| +1$ in $\mathcal{L}_G$,
and this will contradict the maximality of $M$.
Intuitively, $M'$ is obtained by augmenting $M$ with the alternating path $\NewSet \cup \OldSet$, and can be described as follows: We extend $M$ by adding the link-graph edge $\set{ (a,b_i), e_{f(i)}}$.
If $f(i)= 0$, this link-graph edge does not intersect with any edge in $M$, as both $(a,b_i)$ and $e_0$ are not matched in $M$.
Otherwise (i.e., $f(i) \neq 0$), this results in $M$ containing two link-graph edges that touches $e_{f(i)}$.
We fix that by ``swapping'' $(a,b_{f(i)})$ to match with $e_{f^2(i)}$ instead of with $e_{f(i)}$.
We keep applying these upwards swaps until finally adding a link-graph edge
that touches $e_0$, thus strictly increasing the size of $M$.
See Figure \ref{fig:mmprime} for an illustration of $\NewSet$, $\OldSet$, and $M'$.

$\OldSet \subseteq M$ by the correctness of the inductive hypothesis, 
and $\NewSet$ contains only edges from $\mathcal{L}_G$
because each $e_{f(x)}$ links $(a,b_x)$.
Moreover, since $f$ is decreasing, the edges in $\NewSet$ are of the form $((a,b_x),e_{x'})$ with $x' < x$, and therefore $\NewSet \cap M = \emptyset$,
implying that $|M'| = |M|+ 1$.

It remains to show that $M'$ is a matching.
Let $\set{ (a,b_{f^d(i)}),e_{f^{d+1}(i)}} \in \NewSet$.
Since the link-graph edges in $\NewSet$ are vertex disjoint,
it suffices to show that both $(a,b_{f^d(i)})$ and $e_{f^{d+1}(x)}$ do not participate in any other edge in $M \setminus \OldSet$.
For $d = 0$, the pair $(a,b_i)$ is not matched in $M$ according to our assumption. For $d \in [1\dd y-1]$, the pair $(a,b_{f^d(i)})$ is matched in $M$ via the link-graph edge $\set{ (a,b_{f^d(i)}),e_{f^d(i)}} \in M \cap \OldSet$.

It follows that $(a,b_{f^d(i)})$ does not participate in any other edge in $M'$. 
As for $e_{f^{d+1}}$, for $d= y-1$ the edge $e_{f^y(i)} = e_0 = \{a,b_0\}$ is not matched in $M$. For $d\in [0\dd y-2]$, we have $e_{f^{d+1}(i)} = e_x$ for some $x \ge  1$ due to the minimality of $y$. Therefore, the edge $e_x$ is matched in $M$ via the link-graph edge $((a,b_x),e_x) \in M \cap \OldSet$. It follows that $e_x$ does not participate in any other link-graph edge in $M'$, as required.
We see that $M'$ is a matching in $\mathcal{L}_{G}$ of size strictly larger than $M$,
and we have reached a contradiction.

It follows that $(a,b_i)$ must be matched in $M$, concluding the proof of Claim \ref{c:abprimematched}
\end{proof}

With Claim \ref{c:abprimematched}, we have shown how to construct the sequence $v_0,v_1 \ldots v_t$ with the required invariants. In particular, $\set{v_1 \ldots v_t} \subseteq C$ and $|C| \ge t$, which contradicts the minimality of $C$. This concludes the proof of Lemma \ref{lem:sparsifierSimple}.
\end{proof}

We proceed to explain how a similar sparsifier can be constructed for a directed bipartite graph. We start by showing where the undirected construction fails when applied to a directed graph. When assuming to the contrary that $G'$ is not a vertex cut sparsifier, in the proof of Lemma \ref{lem:sparsifierSimple}, we concluded that there is a length $2$ path $a-v-b$ in $G$ that avoids the $(A,B)$ minimum cut $C$ in $G'$. Since $C$ is an $(A,B)$ cut in $G'$, we deduced that one of the edges $\set{a,v},\set{v,b}$ is absent from $G'$. Since the path was undirected, we were able to assume that $e_0 =\set{a,v}$ is the edge missing from $G'$, which leads to the conclusion that $e_0$ is not matched in $M$, even though it can be paired with $(a,b)$ as it links $a$ and $b$.

If the path is directed, we are not able to make this assumption. If $e_0 = (a,v)$ happens to be the edge missing from $G'$, our proof carries in an identical manner and would work for the directed case.  If $(a,v)$ is present in $G'$ and $(v,b)$ is the absent edge - our proof fails. This is due to the fact that $(v,b)$ does not link $a$ and $b$, and therefore can not be used to extend $M$ and reach a contradiction in the proof of Claim \ref{c:abprimematched}. 

We solve this problem by extending the definition of 'linking' edges as follows.
\begin{definition}[Linking Directed Edge]
Let $G=(V,E)$ be a quasi bipartite graph with terminals $T$. For a pair of terminals $a,b \in T$, we say that the edge $e=(a,v) \in E$ is out-linking $a$ and $b$ if $(v,b)\in E$. Similarly, we say that an edge $(v,a)$ is in-linking $a$ and $b$ if $(b,v) \in E$.

The out-link graph and the in-link graph are defined similarly to the link graph, with an edge between $(a,b)$ and $e$ if $e$ out-links (resp. in-links) $a$ and $b$. 
\end{definition}

Now, the foundation of our sparsifier will be two maximum matchings instead of one. A maximum matching $M_{in}$ in the in-link graph and a maximum matching $M_{out}$ in the out-link graph. Our sparsifier consists of the edges $(a,v)$ and $(v,b)$ such that $(a,b)$ was matched to $(a,v)$ in $M_{out}$, or $(v,b)$ was matched to $(a,b)$ in $M_{in}$. We proceed from the problematic point in the undirected case, but with this enhanced construction.

If the edge in the path $a-v-b$ that is absent from $G'$ is $e^{out}_0 = (a,v)$, we get that $e^{out}_0$ was not matched in $M_{out}$ even though it could be matched to $(a,b)$, and the proof carries identically as in the undirected case. If the missing edge is $e^{in}_0 = (v,b)$, we get that $e^{in}_0$ was not matched in $M_{in}$ even though it could be matched to $(a,b)$. From this point on, the proof carries in a symmetrical manner to the proof of the undirected case.

\subsubsection{Generalizations and Applications}

In this section, we show how to generalize our technique to sparsify $\tau$-quasi bipartite graphs.
We prove the following variant of Theorem \ref{t:qbsparsifiergen} for undirected $\tau$ quasi bipartite graphs. A sparsifier for directed $\tau$-quasi bipartite graphs can be obtained by modifying the proof as shown in Section \ref{sec:sparsifierSimple}
\begin{theorem}
Every $\tau$-quasi bipartite graph $G=(V,E)$ with $k$ terminal vertices $T \subseteq V$
admits a vertex-cut sparsifier $G'=(V',E')$ with size bounds 
$|V'| = O(\tau k^2)$ and $|E'| = O((\tau k)^2)$. Furthermore, given $G$ and $T$, 
such a sparsifier can be computed in deterministic polynomial time. 
\end{theorem}
\begin{proof}
We wish to apply a similar construction to the one used in Theorem \ref{t:qbsparsifier}. Let $C_1,C_2\dd C_{\ell}$ be the connected components of $G\setminus T$. We start by shrinking every connected components $C_i$ to create a quasi bipartite graph $G_q = (V_q,T,E_q)$ with $V_q = \{C_i \mid i \in[1\dd \ell] \}$ and $E_q = \{ (t,C_i) \mid \exists_{t\in T, v\in C_i} (t,v) \in E \}$.  

We apply the construction of Theorem \ref{t:qbsparsifier} on $G_q$ to obtain a sparsifier with $G'_q = (V'_q,T,E'_q)$ for $G_q$. We denote as $M$ the maximum matching in $\mathcal{L}_{G_q}$ that was used to construct $G'_q$. We now reverse the shrinking of every connected component to get the sparsifier $G'=(V',T,E')$ with $V' = \cup_{C_i \in V'_q} C_i$ and $E' = \{(t,v)\in E |v\in C_i \textit{ and } (t,C_i) \in E'_q  \} \cup \set{\set{u,v} \mid u,v \in C_i \textit{ for some } i \in [1\dd \ell]}$.

Since $G'_q$ contains $O(k^2)$ edges and $O(k^2)$ vertices, and every vertex in $G'_q$ is extended to a connected component $C_i$ with at most $c$ vertices and $O(c^2)$ edges , we have $|V'| \in O(ck^2)$ and $|E'| = O((ck)^2)$ as required. 

We proceed to show that $G'$ is a vertex cut sparsifier. Assume to the contrary that there are two disjoint sets $A,D \subseteq T$ with $B = T \setminus (A \cup D)$ such that $\mincut_{G'\setminus D}(A,B) < mincut_{G\setminus D}(A,B)$. As in the proof of Theorem \ref{t:qbsparsifier}, we assume that $D = \emptyset$ for the sake of clear presentation. This assumption can be removed without causing any significant change to the proof. Let $C$ be a minimum $(A,B)$ vertex cut in $G'$ that contains a maximal number of terminals. Since $|C| < \mincut_G(A,b)$, $C$ is not a vertex cut in $G$ and we have a path $a,v_1\dd v_p,b$ in $G\setminus C$. Without loss of generality, we assume that $v_1 , v_2 \dd v_p \in C_{i_0}$ for some $i_0\in [1\dd \ell]$, and that the edge $\set{a,v_1}$ is absent from $G'$. It follows that the edge $e_0 = \set{a,C_{i_0}}$ is absent from $G'_q$, and therefore is not matched in $M$. Note that $e_0 = \set{a,C_{i_0}} \in E_q$ links $a$ and $b$ in $G_q$ and therefore can be matched with $(a,b)$ in $M$.

Similarly to the proof of Theorem \ref{t:qbsparsifier}, we use the unmatched edge $e_0$ to construct a sequence of $t = \lceil \frac{|T|}{2} \rceil + 1$ connected components $C_{i_1} \dd C_{i_{t}}$ such that every connected component contains at least one vertex in $C$, thus contradicting its minimality.

Formally, we construct two sequences $C_{i_0},C_{i_1} \dd C_{i_t}$ and $b_1,b_2 \dd b_t$ and a decreasing function $f:[1\dd t] \rightarrow [0\dd t]$ satisfying the following conditions.
\begin{enumerate}
    \item \label{seqinvargen:edges} For every $z\in [0\dd t]$, there is an edge $e_z = \{a,C_{i_z}\} \in E_q$.
    \item \label{seqinvargen:inC} For every $z\in [1\dd t]$, there is a vertex $v_z \in C_{i_z} \cap C$.
    \item \label{seqinvargen:inB} For every $z\in [1\dd t]$, $b_z \in B\setminus C$, and $((a,b_z),e_z) \in M$.
    \item \label{seqinvargen:func} For every $z\in [1\dd t]$, $e_{f(z)}$ links $a$ and $b_z$.
\end{enumerate}

Since $e_0$ is not matched in $M$ and links $a$ and $b$, the pair $(a,b)$ must be matched to another edge $e_1 = (a,C_{i_1}) \in E_q$ that links $a$ and $b$. Since $\set{(a,b) e_1 = \set{a,C_{i_1}}} \in M$, all the vertices of $C_{i_1}$ are present in $G'$, as well as all the edges connecting $a$ or $b$ with vertices in $C_{i_1}$.$C_{i_1}$ is connected, and both $a$ and $b$ are connected to $C_{i_1}$ in $G'$, so there must be a vertex $v_1 \in C_{i_1} \cap C$ for $C$ to disconnect $a$ and $b$.  It follows that $C_{i_0},C_{i_1}$, $b_1$ and $f(1) = 0$ are satisfactory initial assignments for our sequences and for $f$.

The construction of the sequences is carried in an inductive manner. For some $z \ge 2$, assume that we have already constructed the prefixes $C_{i_0},C_{i_1} \dd C_{i_{z-1}}$, $b_1,b_2 \dd b_{z-1}$ and the values $f(x)$ for every $x\in [1\dd z-1]$ in a manner that satisfies our invariants. 

Consider $B' = \{b_1,b_2 \dd b_{z-1}$ and $V' = \{v_1,v_2 \dd v_{z-1}$, and let $C' = C \cup B' \setminus V'$. Since $|C'| = |C|$, and $C'$ contains more terminals than $C$, $C'$ is not an $(A,B)$ vertex cut in $G'$. 

$C'$ can be used to find satisfactory assignment for $C_{i_{z}}$, $b_{z}$ and for $f(z)$ in a similar manner as $C'$ is used in the proof of Theorem \ref{t:qbsparsifier}.
\end{proof}

\subsection{General Graphs with Small Separators}
We lift our result to general graphs with small disconnecting sets by applying the following common observation.

\begin{observation}\label{o:addterminals}
Consider a graph $G=(V,E)$ with terminals $T \subseteq T'\subseteq V$. 
If $G'$ is a vertex-cut sparsifier of $(G,T')$,
i.e., with respect to the extended terminal set $T'$,
then $G'$ is also a vertex-cut sparsifier of $(G,T)$.
\end{observation}

By Observation \ref{o:addterminals}, 
if a graph $G$ with terminals $T$ is not quasi-bipartite 
but has a vertex cover $C$, 
then we can set $T' = T \cup C$ to obtain a quasi-bipartite graph with $|T'| = |C| + k$. 
We can then apply Theorem \ref{t:qbsparsifier} to construct a vertex-cut sparsifier for $(G,T')$.

In order to achieve a polynomial-time construction for the sparsifier, 
we apply a $2$-approximation algorithm to obtain a vertex cover $C'$ with $|C'| \le 2|C|$, 
set $T' = C' \cup T$ and proceed in a similar manner. 
We conclude the above discussion with the following. 
\begin{corollary}
A graph $G=(V,E)$ with terminals $T$ and vertex cover $C$,
where we denote $k = |T|$ and $vc = |C|$, 
admits a vertex cut sparsifier $G'=(V',E')$ with $|E'| = O((k+vc)^2)$. 
Furthermore, $G'$ can be constructed from $G$ in polynomial time.
\end{corollary}

Our vertex-cut sparsifier for $\tau$-quasi-bipartite graphs can be generalized in the same manner. 
If $G=(V,E)$ with terminals $T$ contains a $\tau$-separator $S$, 
we can set $T' = S\cup T$ to obtain a $\tau$-quasi bipartite graph on which Theorem \ref{t:qbsparsifiergen} can be applied.

To achieve a polynomial-time construction algorithm 
we need to efficiently find a small $\tau$-separator. 
We observe that if $G$ has a $\tau$-separator of size $x$, 
then we can find a $\tau$-separator $S'$ with $|S'| \le (\tau +1 )x$ 
by a generalization of the classical $2$-approximation for vertex cover, as follows. 
We initialize an empty $\tau$-separator $S'$, 
and as long as $G\setminus S'$ contains a connected component of size at least $\tau +1$, 
we select a set of $\tau+1$ vertices $V' \in G\setminus S'$ such that $G[V']$ is connected and add $V'$ to $S'$. 
It can be easily verified that for every $V'$ that we select in this process, 
a minimum $\tau$-separator must include at least one vertex in $V'$.

The above discussion yields Corollary \ref{cor:sparsmallsep}. 
\section{Lower Bounds}\label{sec:lower}

In this section, we provide two lower bounds on the size of sparsifiers of undirected quasi-bipartite graphs. Both lower bounds hold even if the sparsifier is only required to maintain minimum cuts between bi-partitions of the terminals,
i.e., the sparsifier $G'$ must satisfy 
\[
  \forall A\subseteq T,\ B=T\setminus A,
  \qquad
  \mincut_{G'}(A,B) = \mincut_{G}(A,B).
\]
(In particular, our lower bounds hold for sparsifiers that satisfy Definition \ref{def:vertexcutsparsifier}.)  

To simplify the exposition, we present our lower bounds in the more general setting of \emph{vertex-weighted} graphs. 
In this setting, vertices have weights given by $w: V\to \NN$,
and a minimum vertex cut is a vertex cut $C\subseteq V$ of minimum total weight
$\sum_{v\in C} w(v)$.
Our lower bounds easily extend to the unweighted setting, 
by replacing each vertex $v$ with $w(v)$ unweighted copies, 
i.e., an independent set $v_1,v_2,\ldots,v_{w(v)}$ of unweighted vertices having the same neighbors that $v$ had. 
If $v$ was a terminal, then all its copies $v_i$ become terminals.
This creates an unweighted graph $G_u=(V_u,E_u)$ with terminals $T_u$ such that 
$|V_u| = \sum_{v\in V}w(v)$ and $|T_u| = \sum_{t\in T}w(t)$. 

It can be easily verified that for every terminal minimum vertex cut $C$ in $G_u$ and every $v\in V \setminus T$, 
either all its copies are in $C$ or none of them, 
i.e., either $\{ v_1,\ldots,v_{w(v)}\} \subseteq C$ or $\{ v_1,\ldots,v_{w(v)}\} \cap C = \emptyset$. 
This means that all the copies of $v$ in $G_u$ act as a unit,
and guarantees that every minimum $(A,B)$-vertex-cut between $A,B \subseteq T$ in $G$ 
is simulated by a minimum $(A_u,B_u)$-vertex-cut in $G_u$, 
where $A_u = \bigcup_{v\in A} \set{v_1,\ldots,v_{w(v)}}$ 
and $B_u$ is defined similarly for $B$,
and vice versa. 

Our results only use small vertex weights, namely at most $4$,
and therefore extend to unweighted graphs with the same asymptotic bounds
on the number of terminals and vertices.

\subsection{Subgraph sparsifiers require $\Omega(k^2)$ vertices}

We start with a lower bound on the number of vertices in a sparsifier
that must be a subgraph of the input graph $G$.
It shows that our sparsifier in Theorem \ref{t:qbsparsifier},
which is a subgraph and has $O(k^2)$ edges (and vertices),
achieves an optimal size bound,
at least when using the technique of subgraph sparsification. 

\begin{theorem}\label{t:subgraphlowerboundweighted}
For every $k\in \NN $, there is a vertex weighted undirected bipartite graph $G_k=(T,N,E)$ with $k$ terminals and $w(v) \le 4$ for all $v\in N$, 
such that every vertex-cut sparsifier of $G$ which is a subgraph of $G$ 
must have $\Omega(k^2)$ vertices. 
\end{theorem}

\begin{proof}
We present a construction for $G_k=(T,N,E)$ for an arbitrary $k\in \NN$ ($N$ is the set of non-terminal vertices). We set $T$ as a union of two sets of $k$ terminals $ A = \{a_i | i \in [1\dd k] \}$ and $D = \{d_i | i\in [1 \dd k] \}$. For every $i\in [1\dd k]$, we connect $a_i$ and $d_i$ with $e_i = \{ a_i,d_i\}$ and set $w(a_i) = 2$ and $w(d_i) = 4$. Finally, for every $\{i,j\} \in \binom{k}{2}$, we add a non terminal $v_{ij}$ connected to $a_i$ and to $a_j$ with $w(v_{ij}) = 1$. For a visualization, see Figure \ref{fig:lowexmaple}

In Section \ref{ap:lemdremove}, we prove the following.

\begin{lemma}\label{l:dremove}
For every $i\in [1\dd k]$ and partition $X, T\setminus X$ of $T$ such that $a_i \in X$ and $d_i \in T\setminus X$, every minimum $(X,T\setminus X)$-cut must contain $a_i$.  
\end{lemma}

We proceed to show that for every $i,j \in {k\choose 2}$, there is a terminals minimum cut that requires the vertex $v_{i,j}$. 
\begin{lemma}\label{l:vijoncut}
Let $i,j \in {k \choose 2}$,
and let $X_{i,j} = \{a_i\} \cup (D\setminus \set{d_j})$.
Then every minimum $(X_{i,j}, T\setminus X_{i,j})$-cut in $G$ contains $v_{i,j}$. 
\end{lemma}

\begin{proof}
Fix $i,j$ and a minimum $(X_{i,j}, T \setminus X_{i,j})$-cut $C$. 
For all $x\neq i,j$, the terminals $d_x,v_x$ are on different sides of the cut,
and thus by Lemma \ref{l:dremove} $C$ must contain $a_x$. 

Now suppose we remove from $G$ the terminals $A' = \{ a_x| x\notin \{i,j\} \}$.
Having deleted these vertices, the path $a_{i}-v_{i,j}-a_{j}$ between $X_{i,j}$ and $T\setminus X_{i,j}$ remains in $G$.  Note that $C_{i,j} = A' \cup \set{v_{i,j}}$ is an $(X_{i,j}, T\setminus X_{i,j})$ vertex cut in $G$ with weight $w(C_{i,j}) = w(A') + 1$. It follows that $C$ may contain one vertex with weight at most $1$ in addition to $A'$, which forces it to include $v_{i,j}$ from the path $a_i - v_{i,j} - a_j$.
The lemma follows. 
\end{proof}

We conclude the proof of Theorem \ref{t:subgraphlowerboundweighted}
by showing that for every $i,j \in {k\choose 2}$, every subgraph vertex sparsifier $G'$ must contain $v_{i,j}$. 
Assume to the contrary that $G'$ is a subgraph of $G_k$ that does not contain $v_{i,j}$. Let $C$ be a minimum $(X_{i,j},T\setminus X_{i,j})$ vertex cut in $G_k$, with $X_{i,j}$ as defined in Lemma \ref{l:vijoncut}. According to Lemma \ref{l:vijoncut}, $v_{i,j} \in C$. Since $G'=(V',E')$ is a subgraph of $G_k$ that does not contain $v_{j,j}$, $C\setminus \{v_{i,j}\} \cap V'$ is an $(X_{i,j},T\setminus X_{i,j})$ vertex cut in $G'$. Since $|C'| < |C|$, we have $micut_{G'}(X_{i,j},T\setminus X_{i,j}) < \mincut_{G_k}(X_{i,j},T\setminus X_{i,j})$,
a contradiction to $G'$ being a minimum vertex-cut sparsifier. 
\end{proof}

\subsection{Sparsifiers require $\tilde{\Omega}(k^2)$ edges} 

In this section, we present our lower bound for the \emph{size} of an arbitrary sparsifier. 
It implies a lower bound on the number of edges in a sparsifier, 
but it has a broader conclusion. 
Informally, we prove that regardless of the method that one uses to represent a graph, a vertex-cut sparsifier of an undirected quasi-bipartite graph requires $\Omega(k^2)$ bits of representation. 
Formally, we prove the following.

\begin{theorem}\label{t:lowergeneral}
For every form of representing graphs using bits, and for every $k\in \NN $, there is a vertex-weighted undirected quasi-bipartite graph $G$ with $k$ terminals, 
such that every vertex-cut sparsifier of $G$ must have $|G'| = \Omega(k^2)$ bits.
\end{theorem}

Here, $|G'|$ stands for the number of bits in the representation of $G'$. 
If we consider a standard representation of a graph as a list of vertices and edges, 
where every edge is represented using $2\log|V|$ bits, we obtain the following.

\begin{corollary}\label{cor:lowerboundedges}
For every $k\in \NN$, there is a vertex-weighted quasi-bipartite graph $G$ with $k$ terminals, for which every vertex-cut sparsifier $G'=(V',E')$ 
must have $|E'| = \Omega(k^2/\log k)$.
\end{corollary}

We define the minimum cut vector of a graph.

\begin{definition}[minimum cut vector]
The minimum cut vector $V_G$ of 
a graph $G=(V,E)$ with terminals $T\subseteq V$
is a vector with $2^{|T|}$ entries. 
The entries of $V_G$ correspond to different subsets $A \subseteq T$,
and their value is $V_G[A] = \mincut_G(A,T\setminus A)$.
\end{definition}

Clearly, if two graphs $G_1$ and $G_2$ have different minimum cut vectors $V_{G_1} \neq V_{G_2}$, then they must have a different vertex cut sparsifier. We prove Theorem \ref{t:lowergeneral} by constructing a large family of graphs with pairwise disjoint minimum cut vectors. Formally, we prove the following.

\begin{lemma}\label{l:graphfamily}
For every $k\in \NN$, there is a family of undirected quasi bipartite graphs $\mathcal{G}_k$ with $|\mathcal{G}_k| = 2^k$ such that every $G\in \mathcal{G}_k$ has $O(k)$ terminals, every two different graphs $G_1,G_2 \in G_k$ have $V_{G_1} \neq V_{G_2}$
\end{lemma}

The existence of $\mathcal{G}_k$ with these properties yields Theorem \ref{t:lowergeneral} via the following reasoning. 
Let $MVC:\mathcal{G}_k \to \mathcal{G}$ be a function that maps every graph to its vertex-cut sparsifier with minimal representation size. 
Since $MVC$ is injective, it has  $|\mathcal{G}_k| = 2^{k^2}$ different output. 
It follows that one of the outputs must be represented using $\Omega(k^2)$ bits.
We are left with the task of proving Lemma \ref{l:graphfamily}.

\begin{proof} 
We start by defining $\mathcal{G}_k$. Consider $G_k=(V,N,E)$ from the proof of Theorem \ref{t:subgraphlowerboundweighted}. $\mathcal{G}_k$ consists of $2^{\Theta(k^2)}$ subgraphs of $G_k$ defined as follows. Recall that $N = \{v_{i,j} | i,j \in { k \choose 2} \}$. For every subset $B \subseteq N$, we define the subgraph $G^B_k$ to be the subgraph of $G_k$ induced by the vertices $T \cup N \setminus B$. We set $\mathcal{G}_k = \bigcup_{B\subseteq N}G^B_k$. Clearly, $|\mathcal{G}_k| = 2^{\Theta(k^2)}$. We proceed to prove that every two graphs in $\mathcal{G}_k$ have different minimum cut vectors.

We make the following claim.
\begin{lemma}\label{l:bdecidescut}
Let $X_{i,j}$ be as in the proof of Lemma \ref{l:vijoncut} 
and let $B \subseteq N$. 
Then $\mincut_{G^B_k}(X_{i,j},T\setminus X_{i,j}) = 2k-4$ if and only if $v_{i,j} \in B$.
\end{lemma}
\begin{proof}
 Due to the same reasoning as in the proof of Lemma \ref{l:dremove}, the minimum $(X_{i,j},T\setminus X_{i,j})$ vertex cut in $G^B_k$ must contain $A'=\{a_x| x\notin \{i ,j\}\}$ with total weight $2k-4$. In $G_k \setminus A'$, $X_{i,k}$ and $T\setminus X_{i,j}$ are connected via the path between $P=a_i,v_{i,j},b_j$, and $A' \cup \set{v_{i,j}}$ is a minimum $(X_{i,j},T\setminus X_{i,j})$ vertex cut.  It follows that if $v_{i,j} \in B$, $v_{i,j}$ is not in $G^B_k$ and $A'$ is a minimum  $(X_{i,j},T\setminus X_{i,j})$ vertex cut in $G^B_k$. Otherwise, if $v_{i,j} \notin B$, $A'$ is not a vertex cut in $G^B_k$. Since every minimum cut must contain $A'$, we have $mincut_{G^B_k}(X_{i,j},T\setminus X_{i,j}) > 2k-4$ in this case.
\end{proof}

We are ready to prove the minimum vertex cut disjointness property of the graphs in $\mathcal{G}_k$. For every two different $B_1,B_2 \subseteq N$, there is at least one vertex $v_{i,j}$ s.t. $v_{i,j} \in B_1$ and $v_{i,j} \notin B_2$ (or vice versa). It follows from Lemma \ref{l:bdecidescut} that $V_{G^{B_1}_k}[X_{i,j}] \neq V_{G^{B_2}_k} [X_{i,j}]$.
\end{proof}

This concludes the proof of Theorem \ref{t:infolowerbound}.
\bibliographystyle{alphaurl}
\bibliography{robi,Sparse}

\appendix

\section{Proof of Lemma \ref{l:dremove}}\label{ap:lemdremove}
\begin{proof}
Let $C \subseteq T \cup V$ be a minimum $(X , T \setminus X)$-cut,
and assume for contradiction that $a_i \notin C$.
Since $a_i$ and $d_i$ are in different sides of the cut and are connected by an edge, we must have $d_i \in C$. Note that removing $a_i$ from $G_k$ disconnects $d_i$ from the rest of the vertices of $G_k$. Therefore, $C^+ = (C \setminus \{ d_i\} ) \cup \{ a_i\}$ is also an $(X, T \setminus X)$-cut. This is a contradiction to the minimality of $C$, as the weight of $C^+$ is at most $|C| - 2$. 
\end{proof}

\section{Figures}
\begin{figure}[htpb]
  \centering
  \scalebox{0.8}{ \input{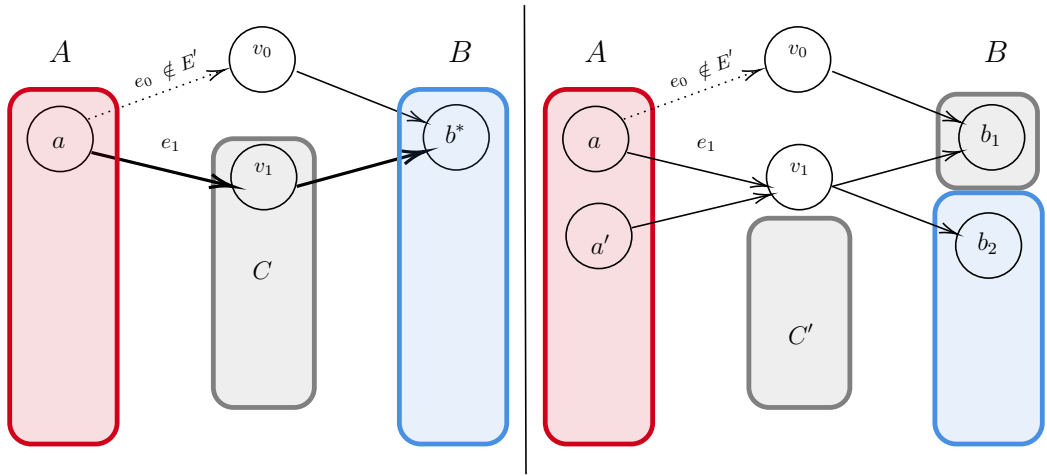}}
  \caption{The first steps of our sparsifier construction.
    The bipartite graph $G$ is illustrated as a tripartite graph,
    where the left and right groups of vertices are the terminal subsets $A$ and $B$,
    and the middle group is the non-terminals.
    \textbf{Left:} The case $i=1$ in the proof of Lemma \ref{lem:sparsifierSimple}.
    $C$ is not an $(A,B)$-vertex-cut in $G$, so there must be a path connecting $a\in A$ and $b_1 \in B$ in $G$ that bypasses $C$ via $e_0 \notin E'$. This suggests the existence of $e_1$ in $G'$ that was matched with $(a,b_1)$.
    \textbf{Right:} The step $i=2$ in our construction. The grey area represents $C'$. Since $C'$ is not an $(A,B)$ cut in $G'$, there is a path from $a'\in A$ to $b_2 \in B \setminus b_1$ in $G'$. Since $C$ \textit{is} an $(A,B)$ cut in $G'$, this path must include $v_1$.}
  \label{fig:step0}
\end{figure}

\begin{figure}
    \centering
    \scalebox{0.8}{ \input{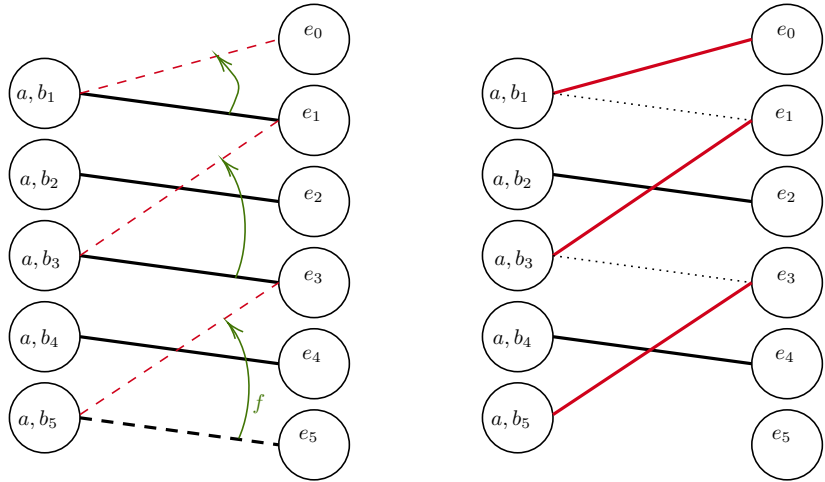}}
    \caption{An illustration of the construction of $M'$ from $M$. \textbf{Left:} The matching $M$ is represented by the bold black edges. Note that $\set{(a,b_5),e_5}$ is not in $M$, as we assume that $(a,b_5)$ is not matched. The red dotted edges are the edges of $\NewSet$, and the green arrows represent the values of $f(5)$, $f(3)$, and $f(1)$. Note that $f$ associates every pair $a,b_x$ with an 'higher' (or 'earlier') edge in the sequence that can be matched to $(a,b_x)$. \textbf{Right:} The matching $M'$, obtained by replacing the edges of $\OldSet$ (dotted thin edges) with the edges of $\NewSet$ (bold, red edges).}
    \label{fig:mmprime}
\end{figure}

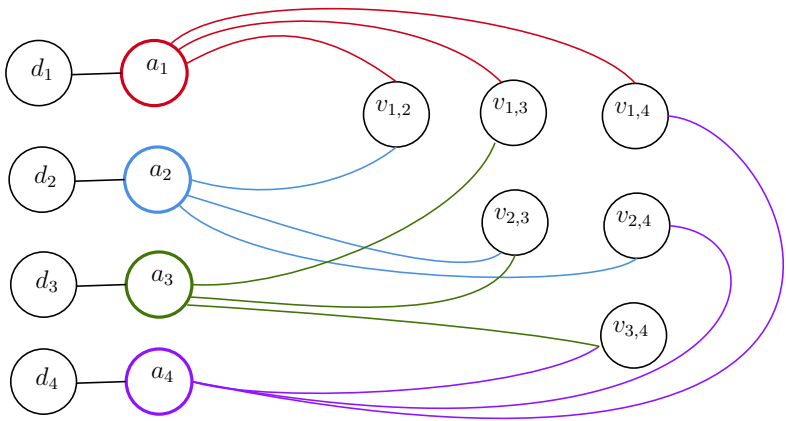
\begin{figure}
    \centering
    \scalebox{0.8}{ \tikzset{every picture/.style={line width=0.75pt}} 

\begin{tikzpicture}[x=0.75pt,y=0.75pt,yscale=-1,xscale=1]

\draw  [color={rgb, 255:red, 208; green, 2; blue, 27 }  ,draw opacity=1 ][line width=1.5]  (153,59.5) .. controls (153,48.18) and (162.18,39) .. (173.5,39) .. controls (184.82,39) and (194,48.18) .. (194,59.5) .. controls (194,70.82) and (184.82,80) .. (173.5,80) .. controls (162.18,80) and (153,70.82) .. (153,59.5) -- cycle ;
\draw  [color={rgb, 255:red, 74; green, 144; blue, 226 }  ,draw opacity=1 ][line width=1.5]  (155,126.5) .. controls (155,115.18) and (164.18,106) .. (175.5,106) .. controls (186.82,106) and (196,115.18) .. (196,126.5) .. controls (196,137.82) and (186.82,147) .. (175.5,147) .. controls (164.18,147) and (155,137.82) .. (155,126.5) -- cycle ;
\draw  [color={rgb, 255:red, 65; green, 117; blue, 5 }  ,draw opacity=1 ][line width=1.5]  (156,192.5) .. controls (156,181.18) and (165.18,172) .. (176.5,172) .. controls (187.82,172) and (197,181.18) .. (197,192.5) .. controls (197,203.82) and (187.82,213) .. (176.5,213) .. controls (165.18,213) and (156,203.82) .. (156,192.5) -- cycle ;
\draw  [color={rgb, 255:red, 144; green, 19; blue, 254 }  ,draw opacity=1 ][line width=1.5]  (156,253.5) .. controls (156,242.18) and (165.18,233) .. (176.5,233) .. controls (187.82,233) and (197,242.18) .. (197,253.5) .. controls (197,264.82) and (187.82,274) .. (176.5,274) .. controls (165.18,274) and (156,264.82) .. (156,253.5) -- cycle ;
\draw   (81,59.5) .. controls (81,48.18) and (90.18,39) .. (101.5,39) .. controls (112.82,39) and (122,48.18) .. (122,59.5) .. controls (122,70.82) and (112.82,80) .. (101.5,80) .. controls (90.18,80) and (81,70.82) .. (81,59.5) -- cycle ;
\draw   (83,126.5) .. controls (83,115.18) and (92.18,106) .. (103.5,106) .. controls (114.82,106) and (124,115.18) .. (124,126.5) .. controls (124,137.82) and (114.82,147) .. (103.5,147) .. controls (92.18,147) and (83,137.82) .. (83,126.5) -- cycle ;
\draw   (84,192.5) .. controls (84,181.18) and (93.18,172) .. (104.5,172) .. controls (115.82,172) and (125,181.18) .. (125,192.5) .. controls (125,203.82) and (115.82,213) .. (104.5,213) .. controls (93.18,213) and (84,203.82) .. (84,192.5) -- cycle ;
\draw   (84,253.5) .. controls (84,242.18) and (93.18,233) .. (104.5,233) .. controls (115.82,233) and (125,242.18) .. (125,253.5) .. controls (125,264.82) and (115.82,274) .. (104.5,274) .. controls (93.18,274) and (84,264.82) .. (84,253.5) -- cycle ;
\draw    (122,60.5) -- (153,59.5) ;
\draw    (124,127.5) -- (155,126.5) ;
\draw    (125,193.5) -- (156,192.5) ;
\draw    (125,254.5) -- (156,253.5) ;
\draw   (304,85.5) .. controls (304,74.18) and (313.18,65) .. (324.5,65) .. controls (335.82,65) and (345,74.18) .. (345,85.5) .. controls (345,96.82) and (335.82,106) .. (324.5,106) .. controls (313.18,106) and (304,96.82) .. (304,85.5) -- cycle ;
\draw   (377,84.5) .. controls (377,73.18) and (386.18,64) .. (397.5,64) .. controls (408.82,64) and (418,73.18) .. (418,84.5) .. controls (418,95.82) and (408.82,105) .. (397.5,105) .. controls (386.18,105) and (377,95.82) .. (377,84.5) -- cycle ;
\draw   (453,86.5) .. controls (453,75.18) and (462.18,66) .. (473.5,66) .. controls (484.82,66) and (494,75.18) .. (494,86.5) .. controls (494,97.82) and (484.82,107) .. (473.5,107) .. controls (462.18,107) and (453,97.82) .. (453,86.5) -- cycle ;
\draw   (378,153.5) .. controls (378,142.18) and (387.18,133) .. (398.5,133) .. controls (409.82,133) and (419,142.18) .. (419,153.5) .. controls (419,164.82) and (409.82,174) .. (398.5,174) .. controls (387.18,174) and (378,164.82) .. (378,153.5) -- cycle ;
\draw   (454,155.5) .. controls (454,144.18) and (463.18,135) .. (474.5,135) .. controls (485.82,135) and (495,144.18) .. (495,155.5) .. controls (495,166.82) and (485.82,176) .. (474.5,176) .. controls (463.18,176) and (454,166.82) .. (454,155.5) -- cycle ;
\draw   (452,224.5) .. controls (452,213.18) and (461.18,204) .. (472.5,204) .. controls (483.82,204) and (493,213.18) .. (493,224.5) .. controls (493,235.82) and (483.82,245) .. (472.5,245) .. controls (461.18,245) and (452,235.82) .. (452,224.5) -- cycle ;
\draw [color={rgb, 255:red, 208; green, 2; blue, 27 }  ,draw opacity=1 ]   (194,53.33) .. controls (239,28.17) and (277,29.33) .. (324.5,65) ;
\draw [color={rgb, 255:red, 208; green, 2; blue, 27 }  ,draw opacity=1 ]   (188,45.33) .. controls (228,15.33) and (347.5,21.67) .. (390,65.33) ;
\draw [color={rgb, 255:red, 208; green, 2; blue, 27 }  ,draw opacity=1 ]   (184,41) .. controls (227,0.33) and (433,20.33) .. (473.5,66) ;
\draw [color={rgb, 255:red, 74; green, 144; blue, 226 }  ,draw opacity=1 ]   (196,126.5) .. controls (246,141.5) and (297.5,128) .. (324.5,106) ;
\draw [color={rgb, 255:red, 74; green, 144; blue, 226 }  ,draw opacity=1 ]   (194,136.33) .. controls (244,151.33) and (363,194.33) .. (390,172.33) ;
\draw [color={rgb, 255:red, 74; green, 144; blue, 226 }  ,draw opacity=1 ]   (189,142.5) .. controls (236,191.33) and (447.5,198) .. (474.5,176) ;
\draw [color={rgb, 255:red, 144; green, 19; blue, 254 }  ,draw opacity=1 ]   (494,86.5) .. controls (576,93.33) and (686,354.33) .. (197,253.5) ;
\draw [color={rgb, 255:red, 144; green, 19; blue, 254 }  ,draw opacity=1 ]   (495,155.5) .. controls (577,162.33) and (548,321.33) .. (197,253.5) ;
\draw [color={rgb, 255:red, 144; green, 19; blue, 254 }  ,draw opacity=1 ]   (197,253.5) .. controls (247,268.5) and (411,260.33) .. (451,231.33) ;
\draw [color={rgb, 255:red, 65; green, 117; blue, 5 }  ,draw opacity=1 ]   (197,192.5) .. controls (262,196.33) and (368,151.33) .. (386,103.33) ;
\draw [color={rgb, 255:red, 65; green, 117; blue, 5 }  ,draw opacity=1 ]   (196,200.33) .. controls (261,204.17) and (380.5,222) .. (398.5,174) ;
\draw [color={rgb, 255:red, 65; green, 117; blue, 5 }  ,draw opacity=1 ]   (194,205.33) .. controls (259,209.17) and (394,219.33) .. (451,231.33) ;

\draw (167,49) node [anchor=north west][inner sep=0.75pt]  [font=\large] [align=left] {$\displaystyle a_{1}$};
\draw (169,116) node [anchor=north west][inner sep=0.75pt]  [font=\large] [align=left] {$\displaystyle a_{2}$};
\draw (170,182) node [anchor=north west][inner sep=0.75pt]  [font=\large] [align=left] {$\displaystyle a_{3}$};
\draw (170,243) node [anchor=north west][inner sep=0.75pt]  [font=\large] [align=left] {$\displaystyle a_{4}$};
\draw (95,49) node [anchor=north west][inner sep=0.75pt]  [font=\large] [align=left] {$\displaystyle d_{1}$};
\draw (97,116) node [anchor=north west][inner sep=0.75pt]  [font=\large] [align=left] {$\displaystyle d_{2}$};
\draw (98,182) node [anchor=north west][inner sep=0.75pt]  [font=\large] [align=left] {$\displaystyle d_{3}$};
\draw (98,243) node [anchor=north west][inner sep=0.75pt]  [font=\large] [align=left] {$\displaystyle d_{4}$};
\draw (310,75) node [anchor=north west][inner sep=0.75pt]  [font=\large] [align=left] {$\displaystyle v_{1,2}$};
\draw (383,74) node [anchor=north west][inner sep=0.75pt]  [font=\large] [align=left] {$\displaystyle v_{1,3}$};
\draw (459,76) node [anchor=north west][inner sep=0.75pt]  [font=\large] [align=left] {$\displaystyle v_{1,4}$};
\draw (384,143) node [anchor=north west][inner sep=0.75pt]  [font=\large] [align=left] {$\displaystyle v_{2,3}$};
\draw (460,145) node [anchor=north west][inner sep=0.75pt]  [font=\large] [align=left] {$\displaystyle v_{2,4}$};
\draw (458,214) node [anchor=north west][inner sep=0.75pt]  [font=\large] [align=left] {$\displaystyle v_{3,4}$};

\end{tikzpicture}}
    \caption{A demonstration of $G_4$. }
\end{figure} \label{fig:lowexmaple}

\end{document}